\documentclass[11pt]{article}
\usepackage[utf8]{inputenc}
\usepackage[T2A]{fontenc}
\usepackage[english]{babel}
\usepackage{amsmath,amssymb,amsthm}
\usepackage[x11names, rgb]{xcolor}
\usepackage{fullpage}
\usepackage{verbatim}
\usepackage{bbold,xspace}
\usepackage{tikz}
\usepackage{verbatim}
\usepackage{relsize}
\usepackage{authblk}

\DeclareMathOperator{\poly}{poly}

\newtheorem{claim}{Claim}
\newtheorem{lemma}{Lemma}
\newtheorem{theorem}{Theorem}
\newtheorem{theoremstar}[theorem]{Theorem}
\newtheorem{lemmastar}[lemma]{Lemma}

\newtheorem{remark}{Remark}

\newtheorem{proposition}{Proposition}

\usepackage{todonotes}

\newcommand{\cO}{\mathcal{O}}

\newcommand{\cOs}{\mathcal{O}^*}
\newcommand{\cL}{\mathcal{L}}

\newcommand{\vc}{\operatorname{vc}}
\newcommand{\tw}{\operatorname{tw}}

\title{Lower Bounds for the Graph Homomorphism Problem\thanks{The research leading to these results has received funding from the Government of the Russian Federation (grant 14.Z50.31.0030), Grant of the President of Russian Federation (МК-6550.2015.1).}}
\author[1,4]{Fedor~V.~Fomin}
\author[2,4]{Alexander~Golovnev}
  \author[4]{Alexander~S.~Kulikov}
\author[3,4]{Ivan~Mihajlin}
 
  \affil[1]{University of Bergen, Norway}
    \affil[2]{New York University, USA}
      \affil[3]{University of California---San Diego, USA}
        \affil[4]{St.~Petersburg Department of Steklov Institute of Mathematics of the Russian Academy of Sciences, Russia}
\date{}

\begin{document}

\maketitle

\begin{abstract}

The graph homomorphism problem (HOM) asks whether the vertices
of a given $n$-vertex graph $G$ can be mapped to the vertices 
of a given $h$-vertex graph $H$ such that each edge of $G$ is mapped to an edge of~$H$. The problem generalizes the graph coloring problem and
at the same time can be viewed as a special case of the $2$-CSP problem.
In this paper, we prove several lower bound for HOM under the Exponential Time Hypothesis (ETH) assumption. The main result is a lower bound
$2^{\Omega\left( \frac{n \log h}{\log \log h}\right)}$. This rules out the existence of a single-exponential algorithm and shows that the trivial upper bound $2^{\cO(n\log{h})}$ is almost asymptotically tight.

We also investigate what properties of graphs $G$ and $H$ make it difficult to solve HOM$(G,H)$. An~easy observation is that an $\cO(h^n)$ upper bound can be improved to $\cO(h^{\vc(G)})$ where $\vc(G)$ is the minimum size of a vertex cover of~$G$. The second lower bound $h^{\Omega(\vc(G))}$ 
shows that the upper bound is asymptotically tight. As to the properties of the ``right-hand side'' graph~$H$, it is known that HOM$(G,H)$ can be solved in time $(f(\Delta(H)))^n$ and $(f(\tw(H)))^n$ where $\Delta(H)$ is the maximum degree of $H$ and $\tw(H)$ is the treewidth of~$H$. This gives single-exponential algorithms for graphs of bounded maximum degree or bounded treewidth. Since the chromatic number $\chi(H)$
does not exceed $\tw(H)$ and $\Delta(H)+1$, it is natural to ask whether
similar upper bounds with respect to $\chi(H)$ can be obtained. We provide a negative answer to this question by establishing a lower bound $(f(\chi(H)))^n$
for any function~$f$. We also observe that similar lower bounds can be obtained for locally injective homomorphisms. 
\end{abstract}


\section{Introduction}\label{sec:intro}
A {\em homomorphism} $G\to H$ from an undirected graph $G$ to an undirected graph $H$ is a mapping
from the vertex set   $G$ to that of $H$ such that the image of every edge of $G$ is an edge of $H$. Then the \textsc{Graph Homomorphism} problem HOM$(G,H)$ is the problem to decide for given graphs $G$ and $H$, whether $G\to H$.
Many combinatorial structures in $G$, for example independent sets and proper vertex colorings,
may be viewed as graph homomorphisms to a particular graph $H$, see the  book of Hell and Ne\v set\v ril~\cite{HellN04}
for a thorough introduction to the topic. 
It was shown by 
Feder and Vardi in \cite{FederV98}  that \textsc{Constraint Satisfaction Problem}  (CSP) can be interpreted  as a homomorphism problem on relational structures, and thus  \textsc{Graph Homomorphism}  encompasses a large family of problems generalizing  \textsc{Coloring}  but less general  than CSP.

Hell and Ne\v set\v ril showed that for any fixed simple graph $H$,
the problem whether there exists a homomorphism from $G$ to $H$  is solvable in
polynomial time if $H$ is bipartite,
and NP-complete if $H$ is not bipartite~\cite{HellN90-On}. 
Since then algorithms and complexity of graph and structures homomorphisms were studied intensively
\cite{Austrin10,BartoKN08,Grohe07,Marx10,Raghavendra08}.

There are two different ways  graph homomorphisms are used to extract useful information about graphs. Let us consider two homomorphisms, from a ``small" graph $F$ into a ``large'' graph $G$ and from a ``large'' graph $G$ into a  ``small" graph $H$, which can be represented by the following formula  (here we borrow  the intuitive description from the  Lov{\'a}sz's book \cite{lovasz2012large})
\[
F\to {\mathlarger{\mathlarger G}} \to H.
\]
Then ``left-homomorphisms" from various small graphs $F$ into $G$ are useful to study the local structure of $G$. For example, if $F$ is a triangle, then the number of  ``left-homomorphisms" from $F$ into $G$ is the number of triangles in graph $G$. 
This type of information is closely related  to sampling, and we refer to the book of Lov{\'a}sz  \cite{lovasz2012large} providing many applications of homomorphisms.  ``Right-homomorphisms" into ``small" different graphs $H$  are related to global observables about graph $G$. 

The trivial brute-force algorithm solving   ``left-homomorphism" from an $f$-vertex graph $F$ into an $n$-vertex graph $G$ runs in time $2^{\cO(f \log{n})}$:
 we try all possible vertex subsets of $G$ of size at most $f$, which is $n^{\cO(f)}$ and then for each subset try all possible $f^f$ mappings into it from $F$.  Interestingly,  this na\"{\i}ve  algorithm is asymptotically optimal.  Indeed, 
as it was shown by 
 Chen et al. \cite{Chen20061346}, assuming Exponential Time Hypothesis (ETH), there is no $g(k)n^{o(k)}$ time algorithm deciding if an input $n$-vertex graph $G$ contains a clique of size at least $k$,  for any computable function $g$.  Since this is a very special case of  \textsc{Graph Homomorphism}    HOM$(F,G)$ with  $F$ being a clique of size $k$, the result of  Chen et al. rules out algorithms for  \textsc{Graph Homomorphism} of running time $g(f)2^{o(f\log{n})}$, from $F$ to $G$, when the number of vertices $f$ in $F$ is significantly smaller than the number of vertices $n$ in $G$. 
 
 Brute-force for  ``right-homomorphism" HOM$(G,H)$, checking all possible mappings from $G$ into $H$, also runs in time $2^{\cO(n \log{h})}$, where $h$ is the number of vertices in $H$.  However, prior to our work there were no results 
 indicating that asymptotically better algorithms, say of running time $2^{\cO({n})}$, are highly unlikely. 
%
%

Our interest in  ``right-homomorphisms"   is due to the recent developments in the area of exact exponential algorithms for \textsc{Coloring} and \textsc{$2$-CSP} problems. 
The area of exact exponential algorithms is about solving intractable problems significantly faster than the trivial exhaustive search, though still in exponential time \cite{FominKratschbook10}. For example, as for  \textsc{Graph Homomorphism}, a na\"{\i}ve  brute-force algorithm for coloring an $n$-vertex  graph $G$ in $h$ colors is to try for every vertex a possible color, resulting in the running time $\cOs(h^n)=2^{\cO(n\log{h})}$.\footnote{$\cOs(\cdot)$ hides polynomial factors in the input length. Most of the algorithms considered in this paper take graphs $G$ and $H$ as an input. By saying that such an algorithm has a running time $\cOs(f(G,H))$ we mean that the running time is upper bounded by
$p(|V(G)|+|E(G)|+|V(H)|+|E(H)|)\cdot f(G,H)$ for a fixed polynomial~$p$.} 
Since $h$ can be of order $\Omega(n)$, the brute-force algorithm computing the chromatic number  runs in time $2^{\cO(n\log{n})}$. It was already observed in 1970s by Lawler   \cite{Lawler76} that the brute-force for  the  \textsc{Coloring}  problem can be  beaten by making use of dynamic programming over maximal independent sets  resulting in single-exponential  running time  $\cOs((1+\sqrt[3]{3})^n)=\cO( 2.45^n)$.
  Almost 30 years later Bj\"{o}rklund, Husfeldt, and Koivisto
 \cite{BjorklundHK2009-Se} succeeded to reduce the running time to $\cOs( 2^n)$.  It is well-known that  \textsc{Coloring} is a special case of graph homomorphism. More precisely, graph $G$ is colored in at most $h$ colors if and only if $G\to K_h$, where $K_h$ is a complete graph on $h$ vertices. Due to this, very often in the literature HOM($G,H$), when $h=|V(H)|\leq n$,  is referred as $H$-coloring of~$G$. And as we observed already,  for $H$-coloring, the brute-force algorithm solving 
 $H$-coloring runs in time $2^{\cO(n\log{h})}$. In spite of  all the similarities between graph coloring and homomorphism, no substantially faster algorithm was known and it was an open question in the area of exact algorithms if there is a single-exponential algorithm solving $H$-coloring in time  $2^{\cO(n+h)}$ \cite{FHK2007,Rzazewski14,Wahlst10,W2011}, see also \cite[Chapter 12]{FominKratschbook10}.

On the other hand, \textsc{Graph Homomorphism} is a special case of \textsc{2-CSP} with $n$  variables and domain of size $h$.  It was shown by Traxler  \cite{T2008} that unless the Exponential Time Hypothesis (ETH) fails, there is no algorithm solving 
\textsc{2-CSP} with $n$ variables and domain of size $h$ in time  $h^{o(n)}=2^{o(n\log h)}$. This excludes (up to ETH)  the existence of a single-exponential $c^n$ time algorithm for some constant $c>1$ for \textsc{2-CSP}.

 \paragraph{Our results.}
In this paper we show that from the algorithmic perspective,   the behavior of  ``right-homomorphism"  is, unfortunately, much closer to  \textsc{2-CSP} than to 
\textsc{Coloring}. The main result of this paper is the following theorem, which 
excludes  (up to ETH) resolvability of  HOM$(G,H)$ in time
$2^{o\left(\frac{n\log h}{\log \log h}\right)}$.


\newsavebox{\boxmainboundone}
\sbox\boxmainboundone{\parbox{\textwidth}{
\begin{theorem}\label{main:theorem_homs} 
Unless ETH fails, for any constant $d>0$ there exists a constant $c=c(d)>0$ such that
for any 
function
$3\le h(n)\le n^d$,
there is no algorithm solving HOM$(G,H)$ for an $n$-vertex graph $G$ and $h(n)$-vertex graph $H$ in time 
\begin{equation}\label{eq:vert}
\cOs\left(2^{\frac{cn\log{h(n)}}{\log\log{h(n)}}} \right) \,.
\end{equation}
\end{theorem}
}}
\noindent\usebox{\boxmainboundone}

\begin{remark}
In order to obtain more general results, in all lower bounds proven in this paper we assume implicitly that the number $h$ of vertices of the graph $H$ is a function of the number $n$ of the vertices of the graph~$G$. At the same time, to exclude some pathological cases we assume that the function $h(n)$ is ``reasonable'' meaning that it is non-decreasing and time-constructible.
\end{remark}

While Theorem~\ref{main:theorem_homs} rules out  the existence of 
a single-exponential algorithm for \textsc{Graph Homomorphism}, single-exponential algorithms can be found in the literature for a number of restricted conditions on  the ``right hand"  graph $H$. For example, when the treewidth of $H$ is at most $t$,  or more generally, when the clique-width of the core of $H$ does not exceed $t$,  the problem is solvable in time  $f(t)^n$ for some function $f$~\cite{W2011}. Another example is when the maximum vertex degree $\Delta(H)$ of $H$ is bounded by a constant. In this case, it is easy to see that a simple branching algorithm also resolves  HOM$(G,H)$  in single-exponential time. Since the chromatic number $\chi(H)$ of $H$ does not exceed 
the treewidth of $H$ and $\Delta(H)$ (plus one), it is natural to ask if a single-exponential algorithm exists when the chromatic number of $H$ is bounded. Unfortunately, this is unlikely to happen.


\newsavebox{\boxchibound}
\sbox\boxchibound{\parbox{\textwidth}{
\begin{theoremstar}\label{chrom:theorem_homs}
Unless ETH fails, 
for any function 
$f \colon \mathbb{N}\to \mathbb{N}$ 
there is no algorithm solving HOM$(G,H)$ for an $n$-vertex graph $G$ and 
a 
graph $H$ in time \(\cOs\left(\left(f(\chi(H))\right)^n\right) \, .\)
\end{theoremstar}
}}
\noindent\usebox{\boxchibound}

Another interesting question about homomorphisms concerns the complexity of the problem when graph $G$ poses a specific structure.
In particular,    when the treewidth of     $G$ does not exceed $t$, then HOM$(G,H)$ is solvable in time $\cOs(h^{t})$ \cite{DiazST02}. 
Let $\vc(G)$ be the minimum size of a vertex cover in graph $G$. We prove that 

\newsavebox{\boxvcbound}
\sbox\boxvcbound{\parbox{\textwidth}{
\begin{theoremstar}\label{vc:theorem_homs}
Unless ETH fails, for any constant $d$ there exists a constant $c=c(d)>0$ such that
for any 
function
$3\le h(n)\le n^d$,
there is no algorithm solving HOM$(G,H)$
for an $n$-vertex graph $G$ and $h(n)$-vertex graph $H$
in time
\(
\cOs\left(h(n)^{c\cdot \vc(G)}\right) \, .
\)

\end{theoremstar}
}}
\noindent\usebox{\boxvcbound}

%
%
 
 Since $\vc(G)$ is always at most the treewidth of $G$, Theorem~\ref{vc:theorem_homs} shows that the known bounds $\cOs(h^{t})=\cOs(h^{\vc(G)})$ on the complexity of homomorphisms from graphs of bounded treewidth and vertex cover  are asymptotically optimal (Note that the minimum vertex cover of $G$ can be found in time 
$1.28^{\vc(G)} \cdot n^{\cO(1)}$~\cite{ChenKX10}). 
 It is   interesting to compare Theorem~\ref{vc:theorem_homs} with existing results on variants of graph homomorphism parameterized by the vertex cover and the treewidth of an input graph. 
 The techniques of obtaining lower bounds developed by Lokshtanov, Marx, and Saurabh in 
 \cite{LokshtanovMS11-superexp}, can be used to show that \textsc{Coloring} 
 cannot be computed in time $2^{o(\vc(G)\log{\vc(G)})}$, unless ETH fails \cite{DLprivat14}. However, the question if  coloring in $h$ colors  of a given graph $G$ can be done  in time  
  $h^{o(\vc(G))}$ remains open.  Another work related to  Theorem~\ref{vc:theorem_homs} is the paper of Marx 
\cite{Marx10} providing lower bounds on the running time of algorithms for ``left-homomorphisms" on classes of structures of bounded treewidth.

\medskip

As a byproduct of our proof of Theorem~\ref{main:theorem_homs}, we obtain   similar lower bounds for  locally injective graph homomorphisms. 
A homomorphism $f \colon  G \to H$ is called \emph{locally injective} if for every vertex 
$u \in V (G)$, its neighborhood is mapped injectively into the neighborhood of $f (u)$ in $H$, i.e., if every two vertices with a common neighbor in $G$ are mapped onto distinct vertices in $H$. 
As graph homomorphism generalizes graph coloring, locally injective graph homomrohism can be seen as a generalization of graph 
 distance constrained labelings. An $L(2, 1)$-labeling of a graph $G$ is a mapping from $V(G)$ into the nonnegative integers such that the 
labels assigned to vertices at distance $2$ are different 
 while labels 
 assigned to adjacent vertices differ by at least $2$. This problem was studied intensively in combinatorics and algorithms, see e.g. Griggs and Yeh \cite{Griggs:1992uq} or  Fiala et al.
 \cite{FialaGK08}. 
 Fiala and  Kratochv\'{\i}l suggested the following generalization of $L(2, 1)$-labeling, we refer  \cite{FialaK08} for the survey. For graphs $G$ and $H$, an $H(2,1)$-labeling is a mapping $f : V(G)\to V(H)$ such that for every pair of distinct adjacent vertices $u,v\in V(G)$, images  $f(u)$ $f(v)$ are distinct and nonadjacent in $H$. Moreover, if the distance between $u$ and $v$ in $G$ is two, then $f(u)\neq f(v)$. It is easy to see that a graph $G$ has an $L(2,1)$-labeling  with maximum label at most $k$ if and only if there is an $H(2,1)$-labeling for $H$ being a $k$-vertex path. Then the following is known, see for example \cite{FialaK08}, there is an $H(2,1)$-labeling of a graph $G$ if and only if there is a locally injective homomorphism from  $G$ to the complement of $H$.
 
 Several single-exponential algorithms for  $L(2,1)$-labeling can be found in the literature, the most recent algorithm is due to  Junosza{-}Szaniawski et al.    \cite{Junosza-SzaniawskiKLRR13} which runs in time $\cO(2.6488^n)$. For $H(2,1)$-labeling, or equivalently for locally injective homomorphisms,   single-exponential algorithms were known only for special cases when the maximum degree of $H$ is bounded  \cite{HavetKKKL11} or  when the bandwidth of the complement of $H$ is bounded \cite{Rzazewski14}. The following theorem explains why no such algorithms were found for arbitrary graph $H$.

\newsavebox{\boxlocalbound}
\sbox\boxlocalbound{\parbox{\textwidth}{
\begin{theoremstar}\label{main:theorem_local_homs}
 Unless ETH fails, for any constant $d>0$ there exists a constant $c=c(d)>0$ such that
for any 
function
$3\le h(n)\le n^d$,
there is no algorithm deciding if there is a  locally injective homomorphism from  an $n$-vertex graph $G$ and $h(n)$-vertex graph $H$ in time 
\( \cOs\left(2^{\frac{cn\log{h(n)}}{\log\log{h(n)}}} \right) \,. \)
\end{theoremstar}
 }}
\noindent\usebox{\boxlocalbound}

 \medskip

  To establish lower bounds for graph homomorhisms,   we proceed in two steps. 
 First we obtain lower bounds for   \textsc{List Graph Homomorphism} by reducing it to the $3$-coloring problem on graphs of bounded degree. More precisely, for a given graph $G$ with vertices of small degrees, we construct an instance $(G',H')$ 
 of  \textsc{List Graph Homomorphism}, such that $G$ is $3$-colorable if and only if there exists a list homomorphism from $G'$ to $H'$. Moreover, our construction guarantees that a ``fast" algorithm for list homomorphism parameterized by the number of vertices, size of a vertex cover or the chromatic number, implies an algorithm for $3$-coloring violating ETH.
 The reduction is based on a  ``grouping" technique, however, to do the required grouping we need a trick exploiting the condition that $G$ has a bounded maximum vertex degree and thus can be colored in a bounded number of colors in polynomial time. In the second step of reductions we proceed from list homomorphisms to normal homomorphisms. Here we need specific  gadgets  with a property that 
  any homomorphism from such a graph to itself preserves an order
  of its specific structures.

%
%
  The remaining part of the paper is organized as follows.
  In Section~\ref{sec:prelimiaries} we give all the necessary definitions.
  Section~\ref{sec:reductions} contains all the necessary reductions which are used to prove lower bounds for the \textsc{Graph Homomorphism} problem in Section~\ref{sec:lowerbounds}. 


 \section{Preliminaries}\label{sec:prelimiaries}

\paragraph{Graphs}
We consider simple undirected graphs, where $V(G)$ denotes the set
of vertices and $E(G)$ denotes the set of edges of a graph $G$.
For a given subset $S$ of $V(G)$, $G[S]$ denotes the subgraph of $G$ induced by $S$,
and $G-S$ denotes the graph $G[V(G)\setminus S]$.
A vertex set $S$ of $G$ is an {\em independent set} if $G[S]$ is a graph
with no edges, and
$S$ is a {\em clique} if $G[S]$ is a complete graph.
The set of neighbors of a vertex $v$ in $G$ is denoted by $N_G(v)$, and the
set of neighbors of a vertex set $S$ is $N_G(S) = \bigcup_{v \in S}N_G(v)
\setminus S$. By $N_G[S]$ we denote the closed neighborhood of the set $S$, i.e., the set $S$ together with all its neighbors: $N_G[S]=S \cup N_G(S)$. For an integer $n$, we use $[n]$ to denote the set of integers $\{1,\dots, n\}$. 

The complete graph on $k$ vertices is denoted by $K_k$.
A {\em coloring} of a graph $G$ is a function assigning a color to each
vertex of $G$ such that adjacent vertices have different colors.
A $k$-coloring of a graph uses at most $k$ colors, and the \emph{chromatic number}   $\chi (G)$ is the smallest number
of colors in a coloring of $G$. 
By Brook's theorem, 
for  any connected  graph $G$ with maximum degree $\Delta>2$,  the chromatic number of $G$ is at most $\Delta$ unless $G$ is a complete graph, in which case the chromatic number is $\Delta + 1$. Moreover, 
a $(\Delta +1)$-coloring of a  graph can be found in polynomial time by a straightforward
 greedy algorithm. 
 
Throughout the paper we implicitly assume that there is a total order on the set of vertices of a given graph. This allows us to treat a $k$-coloring of a $n$-vertex graph simply as a vector in~$[k]^n$.

A~set $S\subseteq V(G)$ is a vertex cover of $G$, if for every edge of $G$ at least one of its endpoints belongs to~$S$.

Let $G$ be an $n$-vertex graph, $1 \le r \le n$ be an integer, and $V(G)=B_1 \sqcup B_2 \sqcup \ldots \sqcup B_{\lceil \frac nr \rceil}$ be a partition of the set of vertices of $G$ into sets of size $r$ with the last set possibly having less than $r$ vertices. Then {\em an edge preserving $r$-grouping} is a graph $G_r$
with vertices $B_1, \ldots, B_{\lceil \frac nr \rceil}$ such that $B_i$ and $B_j$
are adjacent if and only if there exist $u \in B_i$ and $v \in B_j$ such that $\{u,v\} \in E(G)$. To distinguish vertices of the graphs $G$ and $G_r$, the vertices of $G_r$
will be called {\em buckets}.  

For a graph $G$, its {\em square $G^2$} has the same set of vertices as $G$
and $\{u,v\} \in E(G^2)$ if and only if there is a path of length at most $2$ between $u$
and $v$ in $G$ (thus, $E(G) \subseteq E(G^2)$). It is easy to see that if the degree of $G$ is less than $\Delta$ then the degree of $G^2$ is less than $\Delta^2$ and hence a $\Delta^2$-coloring of $G^2$ can be easily found.


 \paragraph{Homomorphisms and list homomorphisms}
 Let $G$ and $H$ be  graphs.  A mapping 
 $\varphi : V(G)\to V(H)$ is a \emph{homomorphism} if for every edge $\{u,v\}\in E(G)$ its image $\{\varphi(u),\varphi(v)\}\in E(H)$.
 If there exists a homomorphism from $G$ to $H$, we often write $G\to H$.
  The \textsc{ Graph Homomorphism} problem  HOM$(G,H)$
  asks whether or not  $G\to H$.

  Assume that for each vertex $v$ of   $G$   we are given a list $\cL(v) \subseteq V (H)$. A \emph{list homomorphism} of $G$ to $H$, also known as  a list $H$-colouring of $G$, with respect to the lists $\cL$, is a homomorphism  $\varphi : V(G)\to V(H)$, such that $\varphi(v) \in \cL(v)$ for all $v\in V (G)$. 
The \textsc{List Graph Homomorphism} problem LIST-HOM$(G,H)$
  asks whether or not  graph $G$ with lists $\cL$ admits a list homomorphism to $H $ with respect to $\cL$.
 
\paragraph{Exponential Time Hypothesis} 
Our lower bounds are based on a well-known complexity hypothesis formulated by  Impagliazzo, Paturi, and Zane   \cite{ImpagliazzoPZ01}.
 
\begin{quote}
\textbf{Exponential Time Hypothesis (ETH)}:  There is a constant $s>0$ such that 3-CNF-SAT with $n$ variables and $m$ clauses cannot be solved in time $2^{sn}(n+m)^{\cO(1)}$.
\end{quote}

This hypothesis is widely applied in the theory of exact exponential algorithms, we refer to \cite{CFKLMPPS2014,LMS2013} for an overview of ETH and its implications. 
 
In our paper we are using the following application  of ETH with respect to \textsc{$3$-Coloring}.
The \textsc{$3$-Coloring} problem is the problem to decide whether the given graph can be properly colored in $3$ colors.

\begin{proposition}[Theorem~$3.2$ in~\cite{LMS2013}, and Exercise $7.27$ in~\cite{S2005}]
\label{prop:col}
Unless ETH fails, there exists a constant $\alpha>0$ such that  \textsc{$3$-Coloring} on $n$-vertex graphs of average degree four cannot be solved in time $\cOs\left(2^{\alpha n} \right)$.
\end{proposition}

It is well known that  \textsc{$3$-Coloring} remains NP-complete on graphs of maximum vertex degree four. Moreover, the classical reduction, see e.g. \cite{GareyJ79}, allows for a given $n$-vertex graph $G$ to construct a graph $G'$ with maximum vertex degree at most four  and   $|V(G')|=\cO(|E(G)|)$ such that $G$ is $3$-colorable if and only if $G'$ is.  Thus Proposition~\ref{prop:col} implies the following (folklore) lemma which will be used in our proofs. 
 
\begin{lemma}\label{lemma:3col}
\label{thm:3col}
Unless ETH fails, there exists a constant $\beta >0$  such that there is no algorithm solving \textsc{$3$-Coloring} on $n$-vertex  graphs of maximum degree four in time $\cOs\left(2^{\beta n} \right)$.
\end{lemma}

\section{Reductions}
\label{sec:reductions}

This section constitutes the main technical part of the paper and contains all the necessary reductions used in the lower bounds proofs. Using these reductions as building
blocks the lower bounds follow from careful calculations. The general pipeline is as follows. To prove a lower bound with respect to  a given graph complexity measure we take a graph $G$ of maximum degree four that needs to be $3$-colored and construct an equisatisfiable instance $(G',H')$
of LIST-HOM (using Lemma~\ref{lemma:3coltolisthom} or Lemma~\ref{lemma:3coltolhomvc}). We then use Lemma~\ref{lemma:lhomtohom} to transform $(G',H')$ into an equisatisfiable instance $(G'',H'')$
of HOM. 
Thus, an algorithm checking whether there exists a homomorphism from $G''$ to $H''$ can be used to check whether the initial graph $G$ can be $3$-colored. At the same time we know a lower bound for \textsc{$3$-Coloring}
under ETH (Lemma~\ref{lemma:3col}). This gives us a lower bound for HOM.
We emphasize that our reductions 
 provide almost tight lower bounds for HOM under ETH.




\newsavebox{\boxlemmatwo}
\sbox\boxlemmatwo{\parbox{\textwidth}{
\begin{lemma}[\textsc{3-Coloring}$(G)$ $\to$ LIST-HOM$(G',H')$ with small $|V(G')|$] 
There exists an algorithm that given an $n$-vertex graph $G$
of maximum degree four and an integer $2 \le r \le n$ constructs an instance $(G',H')$ of LIST-HOM such that 
\( |V(G')|=\lceil n/r \rceil \)  and  \(|V(H')| \le r^{50 r} \) which is satisfiable if and only if the initial graph $G$ is $3$-colorable. The running time of the algorithm is polynomial
in $n$ and the size of the output graphs.
\label{lemma:3coltolisthom}
\end{lemma}
}}
\noindent\usebox{\boxlemmatwo}
\begin{proof}

{\em Constructing $G'$.}
Partition the vertices of $G$ into sets of size $r$ 
(this is possible since $r \le n$) 
arbitrarily and let $G'=G_r$ be
an edge preserving $r$-grouping of $G$ with respect to this partition.
%
%
The maximum vertex degree in graph $G'$ does not exceed $4r$, hence its square can be properly colored with at most $L=16r^2+1$ colors. Fix any such coloring and denote by $\ell(B)$ the color of a bucket $B \in V(G')$.
To distinguish this coloring from a $3$-coloring of $G$ that we are looking for,
in the following we call $\ell(B)$ a {\em label} of~$B$. An important property of this labelling 
is that all the neighbors of any bucket have different labels. Thus to specify a neighbor of a given bucket $B$ it is sufficient to specify the label of this neighbor. 
 This will be crucial for the construction of the graph $H'$ given below.

{\em Constructing $H'$.}
The graph $H'$ is constructed as follows. Roughly, it contains all possible ``configurations'' of buckets from $G'$, where a configuration of $B \in V(G')$ contains its label $\ell(B)$, a $3$-coloring of all $r$ vertices of the bucket $B \subseteq V(G)$, and a $3$-coloring of 
all the neighbors of these $r$ vertices in~$G$. We will use lists to allow mapping of a bucket $B \in V(G')$ to only those configurations that are consistent with a 
3-coloring of the closed neighborhood $N_{G}[B]$.

Formally, a configuration is a tuple
%
\[C=(\ell, c, (p_1,\ell_1,q_1,c_1), \ldots, (p_{4 r},\ell_{4 r},q_{4 r},c_{4 r})) \in [L] \times [3]^r \times ([r] \times [L] \times [r] \times [3])^{4 r} \, .\]

Thus, the number of vertices in $H'$ is equal to (recall that $r \ge 2$)
\begin{equation}\label{eq:rgamma}
  L \cdot 3^r \cdot (r^2 \cdot L \cdot 3)^{4 r} \le r^{4r}\cdot r^{2r}\cdot (r^2\cdot r^7\cdot r^2)^{4r}<r^{50r} \, .
\end{equation}

For a given bucket $B \in V(G')$ such a configuration $C$ sets the following. Integer 
$\ell\in [L]$ is a label of $B$, $c \in [3]^r$ is a 3-coloring of $B \subseteq V(G)$
(recall that we assume a fixed order on the vertices of the graph $G$ so that
the vector $c \in [3]^r$ can be uniquely decoded to a 3-coloring of~$B$).
The rest of $C$ defines a $3$-coloring of all the vertices adjacent to $B$ in~$G$ as follows. Let $\{u_1, v_1\}, \ldots, \{u_k,v_k\} \in E(G)$ be all the edges in the lexicographic order such that $u_i \in B$ and $v_i\not\in B$  for all $i \in [k]$. Note that $k \le 4 r$ since the degree of $G$ is at most $4$ and $|B| \le r$.
Then $(p_i,\ell_i,q_i,c_i) \in [r] \times [L] \times [r] \times [3]$ defines an edge $\{v_i,w_i\}$ and a color of $v_i$ as follows: $p_i \in [r]$ is the number of $u_i$
in~$B$, $\ell_i \in [L]$ is the label of the unique possible neighbor $B'$ of $B$ in~$G'$, $q_i \in [r]$ is the number of $v_i$ in~$B'$, and $c_i \in [3]$ is the color of~$v_i$.

Two configurations $C_1=(\ell^1,c^1,\{(p_i^1,\ell_i^1,q_i^1,c_i^1)\}_{i=1}^{4r})$ and $C_2=(\ell^2,c^2,\{(p_i^2,\ell_i^2,q_i^2,c_i^2)\}_{i=1}^{4r})$ are adjacent if their colorings do not contradict each other. I.e., $C_1$ contains colors of vertices from a bucket labeled by $\ell^2$. We require them to be the same as the ones from the coloring $c^2$ (and similarly for the second configuration).
More formally, $C_1$ and $C_2$ are adjacent if
for every $i \in [4 r]$, if $\ell_i^1=\ell_2$ then $c_i^1$ is equal to the color of $q_i^1$-th vertex in the vector $c_2$, and if $\ell_i^2=\ell_1$ then $c_i^2$
is equal to the color of $q_i^2$-th vertex in $c_1$.

{\em Defining lists of allowed vertices.}
We allow to map a bucket $B \in V(G')$ to a configuration $C=(\ell,\ldots) \in V(H')$ if and only if $\ell(B)=\ell$ and $C$ defines a valid $3$-coloring of $N_G[B]$
(that is, any two adjacent vertices from $N_G[B]$ are given different colors).

{\em Correctness.}
We now show that $G$ is 3-colorable if and only if there is a list-homomorphism from $G'$ to $H'$. The forward direction is clear: given a $3$-coloring of $G$,
one can map each bucket $B$ to the configuration containing the label of this bucket and the coloring of $N_G[B]$. For the reverse direction, we take a homomorphism $\phi \colon G' \to H'$ and for each bucket $B$ we decode from $\phi(B)$ the $3$-coloring of all the vertices of $N_G[B]$. 
Note that if $N_G[B] \cup N_G[B'] \neq \emptyset$ for buckets $B,B' \in V(G')$,
then $\{B,B'\} \in E(G')$. In this case, the edges of $H'$ guarantee that $\phi(B)$
and $\phi(B')$ assign the same color to each vertex in $N_G[B] \cup N_G[B']$.
Hence such a decoding of a $3$-coloring from the homomorphism $\phi$
is well defined. The list constraints of the LIST-HOM instance further guarantee that the resulting $3$-coloring is valid.

{\em Running time of the reduction.}
Clearly, the algorithm takes time polynomial in $n$ and the 
size of the graphs $G'$ and~$H'$.
\end{proof}

\newsavebox{\boxlemmathree}
\sbox\boxlemmathree{\parbox{\textwidth}{
\begin{lemmastar}[\textsc{3-Coloring}$(G)$ $\to$ LIST-HOM$(G',H')$ with small $\vc(G')$]
There exists an algorithm that given an $n$-vertex graph $G$
of maximum degree $4$ and an integer $2 \le r \le n$ constructs an instance $(G',H')$ of LIST-HOM such that 
\( \vc(G')=\lceil n/r \rceil \)  and  \( |V(H')| \le 300^{r} \) which is satisfiable if and only if the initial graph $G$ is $3$-colorable. The running time of the algorithm is polynomial
in $n$ and the size of the output graphs.
\label{lemma:3coltolhomvc}
\end{lemmastar}
}}
\noindent\usebox{\boxlemmathree}
\begin{proof}
The proof is similar to the previous one but is simpler since now we have to guarantee that the vertex cover of $G'$ is small, but not its number of vertices. 

{\em Constructing $G'$.} 
Split the vertices of $G$ into groups of size $r$ arbitrarily. Let $G_r$ be an edge-preserving $r$-grouping with respect to this partition.
Fix a coloring of $G_r$ with $L=5r$ colors.
The graph $G'$ is obtained from $G_r$
by introducing an auxiliary vertex on each edge between two buckets. 
Note that $G'$ is a bipartite graph: 
one part consists of all the buckets while the other one contains all auxiliary vertices.
This in particular implies that its vertex cover is at most $\lceil n/r \rceil$ (the number of buckets).

{\em Constructing $H'$.} The graph $H'$ which
is also bipartite. The left part consists of all possible $L \cdot 3^r$ configurations $(l,c)$ where a configuration is a pair of a $3$-coloring $c \in [3]^r$ of $r$ vertices and a label $l \in [L]$. 
The right part contains pairs of such configurations $(l_1,c_1,l_2,c_2)$ with $l_1 \neq l_2$ and thus has size at most $L^2 \cdot 3^{2r}$. Each such vertex is adjacent to exactly two configurations on the left part~--- to the first component $(l_1,c_1)$ of the pair and to the second one $(l_2,c_2)$. The number of vertices in $H'$ is
\begin{equation}
\label{eq:verth}
L\cdot 3^r + L^2\cdot 3^{2r} = 5r3^r+25r^23^{2r} \le 300^r
\end{equation}
(since $r \ge 2$).

{\em Defining lists.}
Each bucket is allowed to be mapped 
to a configuration $(l,c)$ only if $l$ is the label of this bucket and $c$ is a proper $3$-coloring of its $r$ vertices
(i.e., each vertex of the bucket is assigned a color from its list and any two adjacent vertices are assigned different colors).
An~auxiliary vertex between buckets $B_1$ and $B_2$ of the graph $G'$ is allowed to be mapped to a vertex $(l_1,c_1,l_2,c_2)$ on the right part of $H'$ if and only if $\{l_1,l_2\}=\{l(B_1), l(B_2)\}$
and the colorings $c_1,c_2$ define a proper $3$-coloring of the corresponding $2r$ vertices in~$G$. More precisely, if $l_1$ is the label of $B_1$ then $c_1$ is thought as the coloring of $r$ vertices from $B_1$ and $c_2$ is thought as the coloring of the vertices from $B_2$; if $l_1$ is the label of $B_2$ then the other way around (recall that $l(B_1) \neq l(B_2)$ if there is an auxiliary vertex between $B_1$ and $B_2$ in~$G'$). I.e., labels allow to uniquely decode from a pair of configurations which coloring corresponds to which bucket.

{\em Correctness.}
It is not difficult to see that $G$ has a proper $3$-coloring if and only if there exists a homomorphism from $G'$ to $H'$. Indeed, a $3$-coloring of $G$ can be transformed in a natural way to a homomorphism from $G'$ to $H'$. For this, we map each
bucket of $G'$ to its corresponding configuration (the $3$-coloring of $r$ vertices of the bucket and the label of the bucket). An auxiliary vertex between two buckets $B_1$ and $B_2$ in $G'$ is mapped to the corresponding vertex on the right part of $H'$
(namely, to the vertex consisting of labels of $B_1,B_2$ and $3$-colorings of their $2r$ vertices). 

Conversely, if there is a homomorphism from $G'$ to $H'$ we can decode a $3$-coloring of $G$ from it. To show that this is a proper $3$-coloring of $G$
first note that each vertex $v \in V(G)$ is assigned a color from its list
since $v$ belongs to some bucket $B \in V(G')$ and buckets are allowed to
be mapped only to configurations containing proper colorings of its $r$ vertices. To show that each edge is properly colored consider two adjacent vertices $u,v \in V(G)$. If $u,v$ lie in the same bucket, then the edge $\{u,v\}$ is properly colored by the same reason: this bucket is mapped to a configuration 
containing a proper coloring of all its vertices. If $u \in B$ and $v \in B'$
for different buckets $B,B'$ then these two buckets are adjacent in~$G_r$.
Hence an auxiliary vertex between $B$ and $B'$ in $G'$ guarantees that
$B$ and $B'$ are mapped to configurations containing consistent colorings
of $B,B' \subseteq V(G)$. This in particular means that $\{u,v\} \in E(G)$
is colored properly. 

{\em Running time of the reduction.}
Clearly, the algorithm takes time polynomial in $n$ and the 
size of the graphs $G'$ and~$H'$.
\end{proof}

\newsavebox{\boxlemmafour}
\sbox\boxlemmafour{\parbox{\textwidth}{
\begin{lemmastar}[LIST-HOM $\to$ LIST-HOM with small $\chi(H')$]
\label{lemma:chromatic}
Given an instance $(G,H)$ of LIST-HOM and a $k$-coloring of $G$
one can construct in polynomial time a graph $H'$ such that $\chi(H') \le k$,
$|V(H')|=k|V(H)|$, and $(G,H)$ is equisatisfiable to $(G,H')$.
\end{lemmastar}
}}
\noindent\usebox{\boxlemmafour}
\begin{proof}
Let $c \colon V(G) \to [k]$ be a $k$-coloring of $G$.
The set of vertices of the graph $H'$ is $V(H) \times [k]$, i.e., each vertex of $H'$
is a pair of a vertex of $H$ and a color. Two vertices $(u,i), (v,j) \in V(H')$ are adjacent in $H'$ if and only if $(u,v) \in E(H)$ and $i \neq j$. A vertex $w \in V(G)$ is allowed to
be mapped to a vertex $(u,i) \in V(H')$ (by the list constraints of the instance $(G,H')$) if and only if $w$ is allowed to be mapped to $u \in V(H)$ (by the instance $(G,H)$)
and $c(w)=i$.

It is not difficult to see that the instances $(G,H)$ and $(G,H')$ are equisatisfiable. Indeed, if there is a homomorphism $\phi \colon V(G) \to V(H)$ then $\phi' \colon V(G) \to V(H')$ defined by $\phi'(w)=(\phi(w), c(w))$ is a homomorphism too: if $\{w_1,w_2\} \in E(G)$ then $c(w_1) \neq c(w_2)$ and $\{\phi(w_1), \phi(w_2)\} \in E(H)$ and hence
$\{\phi'(w_1), \phi'(w_2)\} \in E(H')$ (list constraints are also clearly satisfied). The reverse direction is even simpler. If $\phi' \colon V(G) \to V(H')$ is a homomorphism then
set $\phi(w)=u$ such that $\phi'(w)=(u,i)$.

Finally, note that there exists a straightforward $k$-coloring of $H'$:
a vertex $(u,i) \in V(H')$ is assigned the color $i$. 
\end{proof}

\newsavebox{\boxlemmafive}
\sbox\boxlemmafive{\parbox{\textwidth}{
\begin{lemmastar}[LIST-HOM $\to$ HOM]
\label{lemma:lhomtohom}
There is a polynomial-time algorithm that from an instance 
$(G,H)$ of LIST-HOM where $|V(G)|=n$, $|V(H)|=h$, $\chi(H)\le t$ constructs an equisatisfiable instance $(G',H')$ of HOM where $|V(G')| \le n+\Delta$, $\vc(G') \le \vc(G)+\Delta$, $|V(H')| \le \Delta$ for $\Delta=(h+1)(t+11)$, and $\chi(H') \le t+10$.
\end{lemmastar}
}}
\noindent\usebox{\boxlemmafive}
\begin{proof}
{\em Preparations.}
We start from a
simple $6$-vertex gadget $D$ consisting of a $5$-cycle together with an apex vertex adjacent 
to all the vertices of the cycle, see Fig.~\ref{fig:D}. 

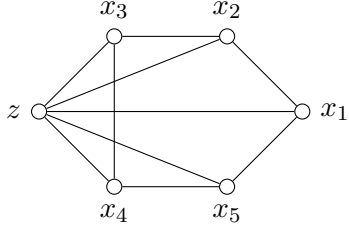
\begin{figure}[ht]
\begin{center}
\begin{tikzpicture}
  \foreach \x/\y/\n/\w in {2/0/1/0, 1/1/2/90, -0.5/1/3/90, -0.5/-1/4/-90, 1/-1/5/-90}
    \node[circle,draw,minimum size=2mm,inner sep=0mm,
    label=\w:{$x_{\n}$}] (b\n) at (\x,\y) {};
  \draw (b1) -- (b2) -- (b3) -- (b4) -- (b5) -- (b1);
  \node[circle,draw,minimum size=2mm,inner sep=0mm,label=180:$z$] (z) at (-1.5,0) {};
  \foreach \n in {1,2,...,5}
    \draw (b\n) -- (z);
\end{tikzpicture}
\end{center}
\caption{The graph~$D$.}\label{fig:D}
\end{figure}

An important property of $D$ is that for each homomorphism $\phi  \colon D \to D$ and $i \in [5]$,
\[ \phi(z) =z \text{ and }\phi(z) \neq \phi(x_i).\]
In words, $z$ is always mapped to $z$ and nothing else is mapped to~$z$. Indeed, because the vertex $z$ is adjacent to all the remaining vertices of~$D$, we have that $\phi(z) \neq \phi(x_i)$. By the same reason, we have that for every $i\in [5]$, $\phi(x_i)\in N_D(\phi(z))$.  But for every $x_i$ its open neighborhood $N_D(x_i)$ induces a bipartite graph. On the other hand, the chromatic number of the cycle $C=x_1x_2x_3x_4x_5$   is three, and thus it cannot be mapped by $\phi$ to 
$N_D(x_i)$ for any $i\in [5]$. Therefore, $\phi(z) =z$.

We join $k$ such $D$'s in a row to construct a larger gadget $T_k$ whose self-homomorphisms preserve the order on $z$'s, see Fig.~\ref{fig:T_k}.
An important property of $T_k$, which will be proven later,  is the following:
for each  $i \in [k]$ and homomorphism $\phi \colon T_k \to T_k$,
$\phi(z_i)=z_i$.
%

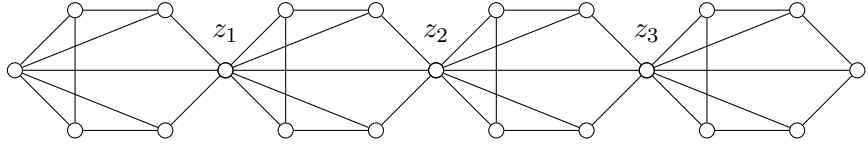
\begin{figure}[ht]
\begin{center}
\begin{tikzpicture}[scale=0.8]
  \foreach \xs in {0, 3.5, 7, 10.5} {

  \begin{scope}
  \foreach \x/\y/\n in {2/0/1, 1/1/2, -0.5/1/3,-0.5/-1/4,1/-1/5}
    \node[circle,draw,minimum size=2mm,inner sep=0mm] (b\n) at (\x+\xs,\y) {};
  \draw (b1) -- (b2) -- (b3) -- (b4) -- (b5) -- (b1);
  \node[circle,draw,minimum size=2mm,inner sep=0mm] (z) at (-1.5+\xs,0) {};
  \foreach \n in {1,2,...,5}
    \draw (b\n) -- (z);
  \end{scope}
  }
  
  \foreach \x/\n in {2/1, 5.5/2, 9/3}
    \node[anchor=south,circle,inner sep=2mm] at (\x,0) {$z_{\n}$};
  
\end{tikzpicture}
\end{center}
\caption{Gadget  $T_k$}\label{fig:T_k}
\end{figure}

We now further extend the graph $T_k$ to increase its chromatic number.
We do this by injecting cliques of size $K_{t+3}$.
This in particular guarantees that it cannot be mapped to a graph with chromatic number less than or equal to $\chi$.

We replace each $z$ in $T_k$ with $K_{t+3}$ and connect every vertex of $K_{t+3}$ to all neighbors of $z$ in the subsequent block. Denote the new graph by $T_{k,t+3}$. See Fig.~\ref{fig:T_kh3}.

\begin{figure}[ht]
\begin{center}
\begin{tikzpicture}[scale=0.7]
  \foreach \xs in {0, 4.1, 8.2, 12.3} {

  \begin{scope}
  \foreach \x/\y/\n in {2/0/1, 1/1/2, -0.5/1/3,-0.5/-1/4,1/-1/5}
    \node[circle,draw,minimum size=2mm,inner sep=0mm] (b\n) at (\x+\xs,\y) {};
  \draw (b1) -- (b2) -- (b3) -- (b4) -- (b5) -- (b1);
  \node[circle,draw,minimum size=12mm,inner sep=0mm] (z) at (-1.45+\xs,0) {$K_{t+3}$};
  \foreach \n in {1,2,...,5}
    \draw (b\n) -- (z);
  \end{scope}
  }
  
  \foreach \x/\n in {2/1, 6.1/2, 10.2/3}
    \node[anchor=south,circle,inner sep=2mm] at (\x-.1,.3) {$z_{\n}$};
  
\end{tikzpicture}
\end{center}
\caption{The gadget $T_{k,t+3}$. An edge from a clique
to a vertex of a cycle means that each vertex of the clique is joined to this vertex.}\label{fig:T_kh3}
\end{figure}
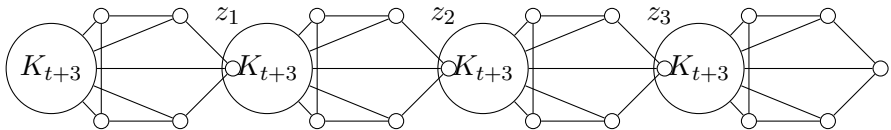

{\em Constructing $G'$.}
Let $A_h$ be a graph consisting of a matching with $h$ edges  $\{ \{a_1, b_1\},\ldots, \{a_h, b_h\}\}$.
Then the graph $G'$  consists of a copy of $G$, a copy of $T_{h,t+3}$, and a copy of $A_h$ with the following additional edges: the vertex $z_i$ from the $i$th block of $T_{h,t+3}$ is adjacent to the vertices $a_i$ and $b_i$. Also we add edges from $G$ to $A_h$: for a vertex $g_i\in G$ we add an edge $\{g_i, a_j\}$ for every $j$, and an edge $\{g_i, b_j\}$ if $j\not\in\mathcal{L}(i)$.  See Fig.~\ref{fig:gprime}. The number of vertices
in $G'$ is at most $n+2h+(h+1)(t+3+5) \le n+(h+1)(t+11)$.

\begin{figure}[ht]
\begin{center}
\begin{tikzpicture}[scale=0.7]
  \foreach \xs/\n in {0/1, 4.1/2, 8.2/3} {

  \begin{scope}
  \foreach \x/\y/\n in {2/0/1, 1/1/2, -0.5/1/3,-0.5/-1/4,1/-1/5}
    \node[circle,draw,minimum size=2mm,inner sep=0mm] (b\n) at (\x+\xs,\y) {};
  \draw (b1) -- (b2) -- (b3) -- (b4) -- (b5) -- (b1);
  \node[circle,draw,minimum size=12mm,inner sep=0mm] (z) at (-1.45+\xs,0) {$K_{t+3}$};
  \foreach \n in {1,2,...,5}
    \draw (b\n) -- (z);
    
  \node[circle,draw,minimum size=2mm,inner sep=0mm,label=left:$a_{\n}$] (p\n) at (1.7+\xs,-2) {};
  \node[circle,draw,minimum size=2mm,inner sep=0mm,label=right:$b_{\n}$] (q\n) at (2.3+\xs,-2) {};
  \draw (b1) -- (p\n) -- (q\n) -- (b1);
  \end{scope}
  }
  
  \foreach \xs in {12.3} {

  \begin{scope}
  \foreach \x/\y/\n in {2/0/1, 1/1/2, -0.5/1/3,-0.5/-1/4,1/-1/5}
    \node[circle,draw,minimum size=2mm,inner sep=0mm] (b\n) at (\x+\xs,\y) {};
  \draw (b1) -- (b2) -- (b3) -- (b4) -- (b5) -- (b1);
  \node[circle,draw,minimum size=12mm,inner sep=0mm] (z) at (-1.45+\xs,0) {$K_{t+3}$};
  \foreach \n in {1,2,...,5}
    \draw (b\n) -- (z);
  \end{scope}
  }
  
  \foreach \x/\n in {2/1, 6.1/2, 10.2/3}
    \node[anchor=south,circle,inner sep=2mm] at (\x,.3) {$z_{\n}$};
    
  \node[circle,draw,minimum size=20mm] (g) at (6.1,-5) {$G$};
  \node[circle,draw,minimum size=2mm,inner sep=0mm,label=left:$i$] (i) at (6.1,-4.4) {};
  
  \draw (i) -- (p1);
  \draw (i) -- (q1);
  \draw (i) -- (p2);
  \draw (i) -- (p3);
  
  
  
\end{tikzpicture}
\end{center}
\caption{The graph~$G'$. A vertex $i \in V(G)$ is connected to $b_j$ if and only if $j \not \in \cL(i)$, where $\cL(i)$ is the list associated with the vertex $i\in V(G).$}\label{fig:gprime}
\end{figure}
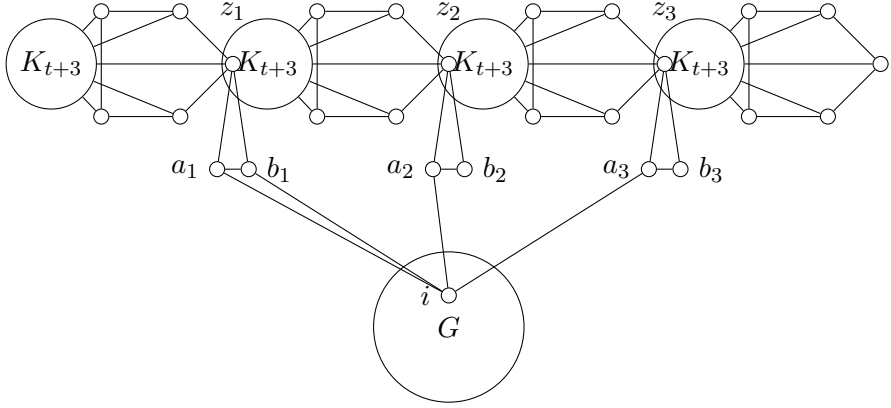

{\em Constructing $H'$.}
The graph $H'$ is constructed similarly. It consists of a copy of $H$, a copy of $T_{h,t+3}$, and a copy of $A_h$. For every $i$ we add edges $\{z_i,a_i\}$ and $\{z_i,b_i\}$ again. Also, each vertex $i$ of $H$ is adjacent to all the vertices from $A_h$ except for $b_i$. See Fig.~\ref{fig:hprime}.
The number of vertices
in $H'$ is at most $h+2h+(h+1)(t+3+5) \le (h+1)(t+11)$.
Now we bound the chromatic number of $H'$. It is easy to see that $t+8$ colors are enough to color $T_{h,t+3}$ (one can color all the cliques and $5$-cycles from left to right). Since $A_h$ is a separator in $H'$, $\chi(H')\le \chi(A_h) +\max(\chi(H), \chi(T_{h,t+3}))\le 2 + \max(t,t+8)=t+10$.

\begin{figure}[ht]
\begin{center}
\begin{tikzpicture}[scale=.7]
  \foreach \xs/\n in {0/1, 4.1/2, 8.2/3} {

  \begin{scope}
  \foreach \x/\y/\n in {2/0/1, 1/1/2, -0.5/1/3,-0.5/-1/4,1/-1/5}
    \node[circle,draw,minimum size=2mm,inner sep=0mm] (b\n) at (\x+\xs,\y) {};
  \draw (b1) -- (b2) -- (b3) -- (b4) -- (b5) -- (b1);
  \node[circle,draw,minimum size=12mm,inner sep=0mm] (z) at (-1.45+\xs,0) {$K_{t+3}$};
  \foreach \n in {1,2,...,5}
    \draw (b\n) -- (z);
    
  \node[circle,draw,minimum size=2mm,inner sep=0mm,label=left:$a_{\n}$] (p\n) at (1.7+\xs,-2) {};
  \node[circle,draw,minimum size=2mm,inner sep=0mm,label=right:$b_{\n}$] (q\n) at (2.3+\xs,-2) {};
  \draw (b1) -- (p\n) -- (q\n) -- (b1);
  \end{scope}
  }
  
  \foreach \xs in {12.3} {

  \begin{scope}
  \foreach \x/\y/\n in {2/0/1, 1/1/2, -0.5/1/3,-0.5/-1/4,1/-1/5}
    \node[circle,draw,minimum size=2mm,inner sep=0mm] (b\n) at (\x+\xs,\y) {};
  \draw (b1) -- (b2) -- (b3) -- (b4) -- (b5) -- (b1);
  \node[circle,draw,minimum size=12mm,inner sep=0mm] (z) at (-1.45+\xs,0) {$K_{t+3}$};
  \foreach \n in {1,2,...,5}
    \draw (b\n) -- (z);
  \end{scope}
  }
  
  \foreach \x/\n in {2/1, 6.1/2, 10.2/3}
    \node[anchor=south,circle,inner sep=2mm] at (\x,.3) {$z_{\n}$};
    
  \node[circle,draw,minimum size=20mm] (g) at (6.1,-5) {$H$};
  \node[circle,draw,minimum size=2mm,inner sep=0mm,label=left:$i$] (i) at (6.1,-4.4) {};
  
  \draw (i) -- (p1);
  \draw (i) -- (q2);
  \draw (i) -- (p2);
  \draw (i) -- (p3);
  \draw (i) -- (q3);
  
  
  
\end{tikzpicture}
\end{center}
\caption{The graph~$H'$. A vertex $i \in V(H)$ is connected to all $a_j$'s and all $b_j$'s except for $b_i$.}
\label{fig:hprime}
\end{figure}
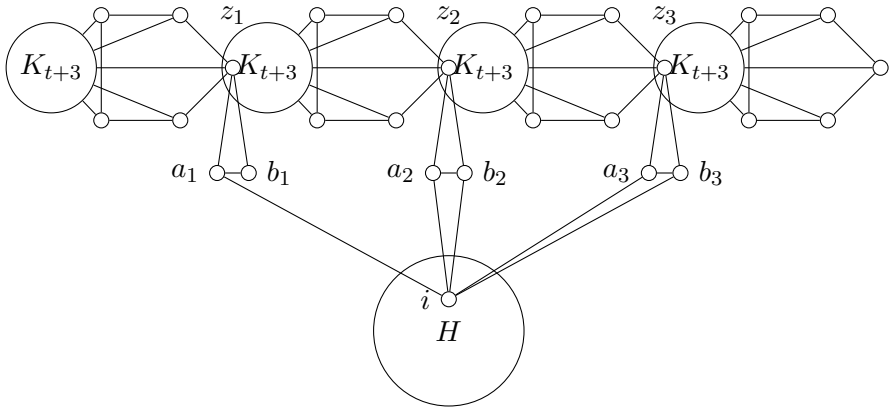

{\em Correctness.} We now turn to prove that the instance $(G,H)$ of LIST-HOM is equisatisfiable to an instance $(G',H')$
of HOM.

\begin{claim}
Any homomorphism $\phi$ from $G'$ to $H'$ maps $T_{h,t+3}$ into $T_{h,t+3}$.
\end{claim}
\begin{proof}[Proof of the claim]
No pair of vertices of the same clique of $T_{h,t+3}$ is mapped to the same vertex in $H'$, because $H'$ has no self-loops. Therefore, cliques from $T_{h,t+3}$ are mapped to cliques from $T_{h,t+3}$ ($H'$ has no more cliques of size $t+3$ since $\chi(H) \le t$). The remaining vertices of $T_{h,t+3}$ have at least $t+3$ neighbors from one clique, therefore they must be mapped to vertices from $T_{h,t+3}$.
\end{proof}

\begin{claim}
Any homomorphism $\phi$ from $G'$ to $H'$ bijectively maps $T_{h,t+3}$ to $T_{h,t+3}$ so that the order of $z$'s is preserved.
\end{claim}
\begin{proof}[Proof of the claim]
By $D_i$ we denote the $i$th block of $T_{h,t+3}$, $D_i$ consists of a clique and a $5$-cycle. Note that two consecutive blocks $D_{i-1}$ and $D_i$ intersect on $z_i$.
\begin{enumerate}
\item {\em Every clique is mapped to a clique.}
First note that a clique is mapped into one block. Indeed, there are no vertices outside of a block that are connected to more than one vertex of the block. Assume, to the contrary, that $K_{t+3}$ is mapped to one block but not to $K_{t+3}$. Then its image has to contain one or two vertices of the $5$-cycle from that block. If the image contains only one vertex of the $5$-cycle, then the image of the $5$-cycle has at most $3$ vertices: one vertex from $K_{t+3}$, two neigbors of the vertex from the $5$-cycle (because all the vertices of the image of the $5$-cycle must be connected to all the vertices of the image of the clique). Note that these three vertices do not form a triangle, therefore the $5$-cycle cannot be mapped to them. If the image of the clique contains two vertices outside of $K_{t+3}$, then for the same reason the image of the $5$-cycle must contain at most $3$ vertices which do not form a triangle. This analysis shows that every $K_{t+3}$ must be mapped to $K_{t+3}$.
\item {\em Every block is mapped to a block.}  We already know that every clique is mapped to a clique. The $5$-cycle from the same block must be mapped to the corresponding $5$-cycle, because it is the only image that contains a cycle of odd length and every vertex of which is connected to the clique (recall that the images of the clique and the cycle do not intersect, since their preimages are joined by edges). Note that since the clique and the cycle are mapped to themselves, $z_i$ has to be mapped to some $z_j$.
\item {\em If $D_i$ is mapped to $D_j$, then $D_{i+1}$ is mapped to $D_{j+1}$.}
The cycle from $D_i$ shares a vertex with the clique from $D_{i+1}$. It is clear that $D_{i+1}$ cannot be mapped into the same block as $D_i$. Indeed, in this case the clique of $D_{i+1}$ would be mapped to a clique in $D_i$ containing $z_i$, but there are no such cliques in $D_i$. Therefore, $D_i$ and $D_{i+1}$ must be mapped in consecutive blocks.
\end{enumerate}
The above proves that for every $i$, $D_i$ is mapped to $D_i$, which implies that any homomorphism preserves the order of $z$'s.
\end{proof}

\begin{claim}
Any homomorphism $\phi$ from $G'$ to $H'$ maps $A_h$ to $A_h$ so that $\{a_i, b_i\}$ is mapped to $\{a_i, b_i\}$.
\end{claim}
\begin{proof}[Proof of the claim]
Every pair $\{a_i, b_i\}$ is connected to $z_i\in T_{h,t+3}$, so it can be mapped either to $\{a_i, b_i\}$ or to some vertices of $T_{h,t+3}$. But in the latter case it would not have paths of length $2$ to all other pairs $\{a_j,b_j\}$.
\end{proof}

\begin{claim}\label{claim4}
Any homomorphism $\phi$ from $G'$ to $H'$ maps $G$ to $H$.
\end{claim}
\begin{proof}[Proof of the claim]
Assume, to the contrary, that a vertex $g\in V(G)$ is mapped to a vertex $v\in V(T_{h,t+3})$ or a vertex $a\in V(A_h)$.
Vertex $g$ is adjacent to at least $h$ vertices from $A_h$, but $v$ and $a$ are adjacent to at most $2$ vertices from $A_h$  (recall that by the previous claim every $\{a_i, b_i\}$ is mapped to $\{a_i, b_i\}$).
\end{proof}

Now we show that the two instances are equisatisfiable. Let $\phi$ be a 
list homomorphism from $G$ to $H$. We show that its natural extension $\phi'$
mapping $T_{h,t+3}$ to $T_{h,t+3}$ and $A_h$ to $A_h$ is a correct homomorphism from $G'$ to $H'$.
This is non-trivial only for edges of $G'$ from $G$ to $A_h$. Consider an edge from
a vertex $i$ of $G$ to the vertex $b_j$. The presence of this edge means that $i$
is not mapped to $j$ by $\phi$. Recall that the $b_j$ is mapped by $\phi$ to $b_j$. This means that the considered edge in $G'$ is mapped to an edge in $H'$ by $\phi'$.

For the reverse direction, let $\phi'$ be a homomorphism from $G'$ to $H'$.
We show that its natural projection is a list homomorphism from $G$ to $H$.
Since $\phi'$ maps $G$ to $H$ it is enough to check that all list constrains are satisfied.
For this, consider a vertex $i$ from $G$ and assume that $j \not \in \mathcal{L}(i)$. Then $\phi'$ does not map $i$ to $j$ as otherwise there would be no image for the edge $\{g_i, b_j\}$, where $g_i$ is the $i$th vertex of $G$.

{\em Running time of the reduction.} The reduction clearly takes time polynomial in the input length.
\end{proof}

\section{Lower bounds for the graph homomorphism problem}\label{sec:lowerbounds}
\subsection{{Parameterization by the number of vertices}}
It is easy to see that the brute-force algorithm solves LIST-HOM$(G,H)$ in time \[ \cOs(h^n)=\cOs(2^{n\log{h}}) \, .\] In this subsection we show a $2^{\Omega\left(\frac{n\log{h}}{\log\log{h}}\right)}$ lower bound under the ETH assumption. 

\noindent\usebox{\boxmainboundone}
\begin{proof}
Let $\gamma>4$ be a large enough constant such that $\frac{\log x}{100 \log \log x}\ge2$ for $x\ge\gamma$. If $h(n)<\gamma$ for all values of $n$, then an algorithm with running time \eqref{eq:vert} would solve \textsc{3-Coloring} in time $\cOs\left(2^{\frac{cn\log{h(n)}}{\log\log{h(n)}}} \right)=\cOs\left(2^{{cn\log{\gamma}}}\right)$ (recall that $h(n) \ge 3$). Therefore, by choosing a small enough constant $c$ such that $c\log{\gamma} < \beta$, we arrive to a contradiction with Lemma~\ref{lemma:3col}.

From now on we assume that $h(n)\ge\gamma$ for large enough values of $n$.
Let $G$ be an $n$-vertex graph of maximum degree $4$ that needs to be $3$-colored. 
We first use Lemma~\ref{lemma:3coltolisthom}, for a parameter $2\le r \le n$ to be defined later, to get an equisatisfiable instance $(G',H')$ of LIST-HOM with $|V(G')|=n/r$ and $|V(H')| \le r^{50 r}$. 
Note that $\chi(H') \le |V(H')|\le r^{50 r}$.
Hence Lemma~\ref{lemma:lhomtohom} provides us with an equisatisfiable instance $(G'',H'')$ of HOM with $|V(G'')| \le n/r+(r^{50r}+1)(r^{50r}+11) \le n/r+r^{102r}$ and $|V(H'')| \le (r^{50r}+1)(r^{50r}+11) \le r^{102 r}$.
Let
\[ r'=\frac{\log n}{204 \log \log n}\,, \quad
r=\min\left(r',\frac{\log h(\frac{2n}{r'})}{102 \log \log h(\frac{2n}{r'})}\right) .\]
Note that $r\le r'<n$. Also,  $h(n)\ge\gamma$ implies that $r\ge2$ for sufficiently large values of $n$. 
Let us show that 
\begin{equation}\label{eq:rr}
r\ge\frac{\log h(\frac{2n}{r'})}{d\cdot204 \log \log h(\frac{2n}{r'})}.
\end{equation}
This clearly holds if $r<r'$, so consider the case $r = r'$. The function $\log x/\log \log x$ increases for $x>4$. Recall that $h(n) \ge \gamma$ for large enough values of~$n$, hence $h(2n/r') \ge \gamma > 4$ for large
enough values of~$n$. Hence
\[ \frac{
     \log \left( \frac{2n}{r'} \right)
   }{
     \log \log \left( \frac{2n}{r'} \right)
   } 
   \le
   \frac{d \log n}{\log \log n + \log d} \le \frac{d\log n}{\log \log n} = 204dr'=204dr
\]
which implies~\eqref{eq:rr}. 
Then
\[ |V(G'')| \le \frac nr + r^{102 r} \le \frac nr + (\log n) ^{\frac{\log n}{2\log \log n}} = \frac nr + \sqrt{n} \le \frac{2n}{r} \, ,
\]
 \[ |V(H'')| \le r^{102 r}  \le \left(\log h\left(\frac{2n}{r'}\right)\right)^{\frac{\log h(\frac{2n}{r'})}{\log \log h(\frac{2n}{r'})}} = h\left(\frac{2n}{r'}\right) \le  h\left(\frac{2n}{r}\right) 
 \le h(|V(G'')|).
\]
Hence one can add isolated vertices to both $G''$ and $H''$ (clearly this does not change the problem) such that $|V(G'')|=2n/r$ and $|V(H'')|=h(2n/r)$ and run an algorithm from the theorem statement on the instance $(G'',H'')$.
  
Note that the running time of the reduction is \[ 
\poly(|G|, |G'|, |G''|, |H|, |H'|, |H''|)=\poly(n,h(2n/r))=\cOs(1) \, . \]
Thus, an algorithm with running time \eqref{eq:vert} for HOM implies an algorithm for \textsc{3-Coloring} with running  time
\[\cOs\left( 2^{c\cdot \frac{2n}{r} \cdot \frac{\log h(\frac{2n}{r'})}{\log \log h(\frac{2n}{r'})}}  \right) = \cOs \left( 2^{408cdn}\right) \]
(recall the inequality \eqref{eq:rr}).
Therefore, by choosing a small enough constant $c>0$ such that $408cd < \beta$, we arrive to a contradiction with Lemma~\ref{lemma:3col}.
\end{proof}


\subsection{Parameterization by the chromatic number of $H$}
\noindent\usebox{\boxchibound}
\begin{proof}
Let $G$ be a graph on $n$ vertices of degree at most four that needs to be $3$-colored. 
Let $2 \le r \le n$ be a constant to be chosen later. 
We first greedily find a $5$-coloring of $G$. Then add at most $5r$ isolated vertices to the graph
and assign each of them one of five colors such that in the resulting $5$-colored graph the number of vertices of each color is a multiple of~$r$. Now, partition the vertices into groups of size $r$ such that each group consists of vertices of the same color. Then construct an equisatisfiable instance $(G_1,H_1)$ of LIST-HOM according to this partition using Lemma~\ref{lemma:3coltolisthom}. Since $G_1$ is an edge preserving  $r$-grouping of $G$ and each group contains only vertices of the same color, we conclude that $\chi(G_1) \le 5$
(each bucket can be assigned the color of its $r$ vertices). We have that $|V(G_1)| \le \lceil n/r \rceil +5$ and $|V(H_1)| \le r^{50r}$.

We now apply Lemma~\ref{lemma:chromatic} to construct a graph $H_2$ such that $\chi(H_2) \le 5$, $|V(H_2)| \le 5r^{50r}$, and $(G_1,H_1)$ is equisatisfiable to $(G_1,H_2)$.

Finally, we use Lemma~\ref{lemma:lhomtohom} to construct an instance $(G_3,H_3)$ of HOM that is equisatisfiable to $(G_1,H_2)$. Then 
\[ |V(G_3)| \le \lceil n/r \rceil +5+(5r^{50r}+1)(5+1) \le 2n/r\]
for large enough $n$ since $r=O(1)$. Also, $\chi(H_3) \le 15$. 

%
%

Thus, if HOM$(G,H)$ could be solved in time 
$\cOs\left( f(\chi(H))^{|V(G)|} \right)$ then such an algorithm
could be used to solve \textsc{3-Coloring} for $G$ in time
 \[ \cOs\left( f(\chi(H_3))^{|V(G_3)|} \right) =\cOs\left( f(15)^{2n/r}\right) \, .\]
 Thus, for a large enough constant $r$ such that $f(15)^{2/r} < 2^\beta$ we get a contradiction with Lemma~\ref{lemma:3col}.
%
\end{proof}
%
%

\subsection{{Parameterization by the vertex cover of~$G$}}
\label{sec:vc_lower_bounds}
The following lemma follows from known results about graph homomorphism on graphs of bounded treewidth (the minimum vertex cover of a graph is always at least its treewidth), see e.g.~\cite{DiazST02}. For vertex cover parameterization such an upper bound becomes very simple and we add the proof for completeness. Note that the minimum vertex cover of $G$ can be found in time 
$1.28^{\vc(G)} \cdot n^{\cO(1)}$~\cite{ChenKX10}.

\begin{lemma}\label{lm:vcup}
There exists an algorithm solving LIST-HOM$(G,H)$
for an $n$-vertex graph $G$, $h$-vertex graph $H$, and a vertex cover $C \subseteq V(G)$ of $G$ in time $\cOs\left( h^{|C|}\right)$.
\end{lemma}

\begin{proof}
The algorithm just goes through all possible $h^{|C|}$ mappings
of the vertices from $C$ to the vertices of $H$. For each such  mapping $\phi$
 it is easy to check whether  $\phi$ can be extended to a 
homomorphism from $G$ to~$H$. 
Indeed,  because $V(G) \setminus C$ is an independent set, a mapping $\phi$
can be extended to a homomorphism if and only if 
(a) $\phi$ is a homomorphism from $G[C]$ to $H$, and (b)  
 for every $v\in V(G) \setminus C$ there is $u\in V(H)$ such that $u\in L(v)$ and 
 the neighbourhood  $N_H(u)$ contains all images in $\phi$ of neighbours of $v$.
 
 Both properties can be clearly checked in time polynomial in  the input length.
\end{proof}
 
Below we show that this simple upper bound is unlikely to be
substantially improved.
\noindent\usebox{\boxvcbound}
\begin{proof}
Let $G$ be an $n$-vertex graph of maximum degree $4$ that needs to be $3$-colored.
Let $\gamma$ be a large enough constant such that
\begin{equation}
\label{eq:gammavc}
\frac{\log x}{40} \ge 2
\end{equation}
for all $x > \gamma$.
First consider the case when $h(n) < \gamma$ for all~$n$.
Then, since $h(n) \ge 3$ for all $n$, an algorithm 
with the running time
$\cOs\left(h(n)^{c\cdot \vc(G)}\right)$.
would solve \textsc{3-Coloring} in time 
\[ \cOs\left(h(n)^{c \cdot \vc(G)}\right)= \cOs\left(\gamma^{cn}\right) \, .\]
Then for a small enough constant $c$ such that $\gamma^c < \beta$ one gets
a contradiction with Lemma~\ref{lemma:3col}. Thus, in the following we assume that $h(n) \ge \gamma$
for large enough values of~$n$.

We now use Lemma~\ref{lemma:3coltolhomvc} to construct an equisatisfiable instance $(G',H')$
of LIST-HOM with $\vc(G')=n/r$ and $|V(H')| \le 300^r$ for a parameter $2 \le r \le n$ to be chosen later. 
We then use Lemma~\ref{lemma:lhomtohom}
to construct an equisatisfiable instance $(G'',H'')$ such that
\[
\vc(G'') \le \vc(G')+(|V(H')|+1)(\chi(H')+11) \le n/r+(300^r+1)(300^r+11) \le n/r + 2^{20r}\,,
\]
\[ V(H'') \le (|V(H')|+1)(\chi(H')+11) \le 2^{20r} \, .\]

Now set
\[r'= \frac{\log n}{40}\,, \quad 
r = \min \left( r', \frac{\log h\left( \frac{2n}{r'}\right)}{40 
}\right) \, .\]
From \eqref{eq:gammavc} it follows that $r \ge 2$ for large enough~$n$. It is also clear that $r \le n$ for large enough~$n$.
Then \[\vc(G'') \le n/r+2^{20r} \le n/r+2^{20r'}=n/r+2^{\log n/2}=n/r+\sqrt{n} < 2n/r\] for large enough~$n$. 
Also, \[|V(H'')| \le 2^{20r} \le h(2n/r')^{1/2} \le h(2n/r') \le h(2n/r) \le h(\vc(G'')) \le h(|V(G'')|)\, .\]
Hence one can use an algorithm from the theorem statement for an instance $(G'',H'')$.

We now show that 
\begin{equation}
\label{eq:vca}
r \ge \frac{\log h\left( \frac{2n}{r'} \right)}{40d} \,.
\end{equation}
This is clearly true if $r < r'$. Let now $r=r'$. Then
\[ \log h\left( \frac{2n}{r'} \right) \le \log h(n) \le d\log n = 40dr' = 4dr \]
which implies~\eqref{eq:vca}.
Using this inequality, we conclude that an algorithm solving HOM in  
$\cOs\left(h(n)^{c\cdot \vc(G)}\right)$
allows to solve \textsc{3-Coloring} in time
\[\cOs\left(|V(H'')|^{\vc(G'')}\right) \le 
  \cOs\left( h\left( \frac{2n}{r'} \right)^{2cn/r}\right)=   
  \cOs\left( 2^{2cn\log h\left( \frac{2n}{r'} \right)/r}  \right)  \le 
  \cOs\left( 2^{80cdn}\right) \, .
   \]
Thus, for a small enough constant $c$ such that $2^{80cd} < \beta$ we get a contradiction with Lemma~\ref{lemma:3col}.
\end{proof}

\subsection{Locally injective homomorphisms}
\noindent\usebox{\boxlocalbound}
\begin{proof}
The proof 
is almost identical to the proof of Theorem~\ref{main:theorem_homs}.

Let us  observe that in the reduction  
in Lemma~\ref{lemma:3coltolisthom},  in graph $G'$, we take a coloring (in the proof we refer to such coloring as to labeling) of the square of $G'$. Thus 
for every bucket $v$ of $G'$, all its neighbors are  labeled by different colors. The way we construct the lists, 
only buckets with the same labels can be mapped to the same vertex of $H'$. Thus 
for every vertex $v$ of $G'$, no pair of its neighbors can be mapped to the same vertex. Hence every list homomorphism from $G'$ to $H'$ is locally injective. Therefore the result of Lemma~\ref{lemma:3coltolisthom} holds for locally injective list homomorphisms as well and we obtain the following lemma.

\begin{lemma}\label{lemma:3-local-col-list}
There exists an algorithm that given an $n$-vertex graph $G$
of maximum degree four and an integer $2 \le r \le n$ constructs a pair of graphs  $G'$ and $H'$ with
\[ |V(G')|=\lceil n/r \rceil \text{ and } |V(H')| \le r^{50 r} \]  such that there is a locally injective list homomorphism from $G'$  to $H'$  if and only if $G$ is $3$-colorable. The running time of the algorithm is polynomial
in $n$ and the size of the output graphs.
\end{lemma}

In the reduction of  Lemma~\ref{lemma:lhomtohom}, we established that every homomorphism from $G'$ to $H'$
maps   $T_{h,t+3}$ to $T_{h,t+3}$ and $A_h$ to $A_h$ so that $\{a_i, b_i\}$ is mapped to $\{a_i, b_i\}$. 
Thus for vertices of these structures, every homomorphism is locally injective. 
By Claim~\ref{claim4}, 
any homomorphism $\phi$ from $G'$ to $H'$ maps $G$ to $H$. Therefore there  is a locally injective homomorphism from  $G'$ to $H'$ if and only if there is a  locally injective list homomorphism from $G$ to $H$.  Then by making use of Lemma~\ref{lemma:3-local-col-list},  the calculations performed in the proof of Theorem~\ref{main:theorem_homs} we conclude with the proof of the theorem.
\end{proof}
\section{Conclusion}\label{sec:conclusion}
 We conclude with several open problems around graph homomorphisms. 
\begin{itemize}
\item The first natural question is if our bounds are tight. For example, can the bound in
Theorem~\ref{main:theorem_homs}  be improved to match asymptotically the $2^{\cO(n\log{h})}$ running time of the brute-force algorithm?
 On the other hand, there is no argument ruling out the possibility of solving the problem in time $2^{o(n\log{h})}$. 
\item The second question is due to Daniel Lokshtanov \cite{DLprivat14}:
Is it possible to color a graph $G$ in $h$ colors in time $h^{o(\vc(G))}$.
\item Deciding if an $n$-vertex graph $F$ is a subgraph of an $n$-vertex graph 
$G$  can be done in time $2^{\cO(n\log{n})}$ by trying all possible vertex permutations of both graphs.
Can this problem be solved in time $2^{o(n\log{n})}$?
\end{itemize}
\paragraph{Acknowledgement} We are grateful to Daniel Lokshtanov and Saket Saurabh for helpful discussions. 
\vspace{-0.2in}
\bibliographystyle{plain}
\bibliography{hom}

\end{document}